\documentclass[11pt]{article}

\usepackage{references}
\usepackage{cite,graphicx,verbatim}
\usepackage{fullpage}

\newtheorem{theorem}{Theorem}[section]
\newtheorem{lemma}{Lemma}[section]
\newtheorem{corollary}{Corollary}[section]
\newtheorem{conjecture}{Conjecture}[section]

\newcommand{\classNP}{{\sf NP}}
\newcommand{\mod}{\mbox{ mod }}

\newcommand{\sq}{\hbox{\rlap{$\sqcap$}$\sqcup$}}
\newcommand{\qed}{\hspace*{\fill}\sq}
\newenvironment{proof}{\noindent {\bf Proof.}\ }{\par\vskip 4mm\par}

\begin{document}


\title{Network Creation Games with Disconnected Equilibria\thanks{A
    preliminary version of this will appear in the Proceedings of
    the 4th Workshop on Internet \& Network Economics, Shanghai,
    China, 2008.}}

\author{Ulrik Brandes\thanks{Department of Computer Science, University of
    Konstanz, Germany, {\tt ulrik.brandes@uni-konstanz.de}} \and Martin
  Hoefer\thanks{Supported by DFG-Graduiertenkolleg
    ``AlgoSyn''. Lehrstuhl Informatik I, RWTH Aachen University, Germany, {\tt
      mhoefer@cs.rwth-aachen.de}} \and Bobo Nick\thanks{Department of
    Computer Science, University of Konstanz, Germany, {\tt
      bobo.nick@uni-konstanz.de}}}

\date{}

\maketitle 


\begin{abstract}
  In this paper we extend a popular non-cooperative network creation
  game (NCG)~\cite{Fabri03} to allow for disconnected equilibrium
  networks. There are $n$ players, each is a vertex in a graph, and a
  strategy is a subset of players to build edges to. For each edge a
  player must pay a cost $\alpha$, and the individual cost for a
  player represents a trade-off between edge costs and shortest path
  lengths to all other players. We extend the model to a
  \emph{penalized game} (PCG), for which we reduce the penalty counted
  towards the individual cost for a pair of disconnected players to a
  finite value $\beta$. Our analysis concentrates on existence,
  structure, and cost of disconnected Nash and strong
  equilibria. Although the PCG is not a potential game, pure Nash
  equilibria always and pure strong equilibria very often exist. We
  provide tight conditions under which disconnected Nash (strong)
  equilibria can evolve. Components of these equilibria must be Nash
  (strong) equilibria of a smaller NCG. However, in contrast to the
  NCG, for almost all parameter values no tree is a stable
  component. Finally, we present a detailed characterization of the
  price of anarchy that reveals cases in which the price of anarchy is
  $\Theta(n)$ and thus several orders of magnitude larger than in the
  NCG. Perhaps surprisingly, the strong price of anarchy increases to
  at most 4. This indicates that global communication and coordination
  can be extremely valuable to overcome socially inferior topologies
  in distributed selfish network design.
\end{abstract}


\section{Introduction}

Networks are ubiquitous in modern society. It is therefore not
surprising that the study of network creation has attracted much
research interest from various disciplines. In recent years, it has
been understood that the distributed formation of networks may be
subject to economic considerations. In particular, the creation of
social, economic, and computational networks was formulated as a game
with selfish agents. A general framework for such an approach was
proposed by Jackson and Wolinsky~\cite{Jackson96}. In their games
there are $n$ players and each player is a vertex in a graph. A
strategy consists of choosing which incident edges to build. Depending
on the network structure there is a payoff for each player, and
players adjust their strategy to maximize their payoff. A general
finding was that there are games, in which no efficient network is
stable for a concept of pairwise stability, which requires bilateral
consent to construct a connection. The extensions and adjustments to
this model are numerous~\cite{Jackson04}. In particular, several works
extended the model to unilateral link creation and the Nash
equilibrium as stability concept~\cite{Bala00,Dutta00}.

A particularly interesting variant was proposed in the context of
distributed systems and the Internet by Fabrikant et
al.~\cite{Fabri03}. In their network creation game (NCG) the cost of
creating an edge is fixed to a parameter $\alpha$. Edge creation is
unilateral, and the cost for a player is a trade-off between the costs
for created edges and the structural position in the network measured
by shortest path distances to all other players. In~\cite{Fabri03} and
consecutive work~\cite{Albers06,Demaine07} the inefficiency of Nash
equilibria was quantified using the~\emph{price of anarchy}, which
captures the deterioration in cost of the worst Nash equilibrium
against a social optimal state. The presently known results on the
price of anarchy are summarized in Figure~\ref{fig:NCG}.
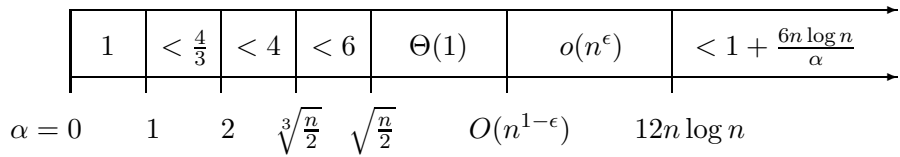
\begin{figure}[h]
\begin{center}
\setlength{\unitlength}{1mm}
\setlength{\linewidth}{20mm}
\begin{picture}(122,19)
\put(2,2){\normalsize $\alpha = 0$}
\put(10,10){\vector(1,0){110}}
\put(10,19){\vector(1,0){110}}

\put(10,8){\line(0,1){11}}
\put(14,13){\normalsize 1}
\put(20,8){\line(0,1){11}}
\put(20,2){\normalsize 1}
\put(22,13){\normalsize $< \frac{4}{3}$}
\put(30,8){\line(0,1){11}}
\put(30,2){\normalsize 2}
\put(32,13){\normalsize $< 4$}
\put(40,8){\line(0,1){11}}
\put(37,2){\normalsize $\sqrt[3]{\frac{n}{2}}$}
\put(42,13){\normalsize $< 6$}
\put(50,8){\line(0,1){11}}
\put(47,2){\normalsize $\sqrt{\frac{n}{2}}$}
\put(55,13){\normalsize $\Theta(1)$}
\put(68,8){\line(0,1){11}}
\put(63,2){\normalsize $O(n^{1-\epsilon})$}
\put(75,13){\normalsize $o(n^\epsilon)$}
\put(90,8){\line(0,1){11}}
\put(85,2){\normalsize $12n \log n$}
\put(93,13){\normalsize $< 1 + \frac{6n \log n}{\alpha}$}
\end{picture}
\end{center}
\caption{\label{fig:NCG} Price of anarchy in the NCG}
\end{figure}
Other equilibrium concepts were also studied, e.g.\ existence and cost
of pairwise stable equilibria~\cite{Corbo05}, or of strong
equilibria~\cite{Andelman07}. Extensions to more general edge costs or
different player cost trade-offs proved useful in the analysis of
mobile peer-to-peer
networks~\cite{Eidenbenz03,Moscibroda06,Ackermann07}.

In network analysis~\cite{Brandes05}, the inverse of the sum of
shortest path lengths is one of the most commonly used measures of
centrality known as \emph{closeness}~\cite{Freeman79}. A problem with
closeness is that global connectivity is required for the scores to be
comparable. This means that in the NCG for moderate to high edge
costs the trade-off is distorted by the enforcement of
connectivity. Thus, it was not surprising that trees proved to be a
prominent equilibrium structure~\cite{Fabri03}.

In this paper, we remedy this problem by replacing the infinite cost
of not being connected by a finite penalty $\beta$. This corresponds
directly to a variant of closeness centrality proposed by Botafogo et
al.~\cite{Botafogo94}. This is also closely related to a measure
termed \emph{radiality}~\cite{Valente98}, although here $\beta$
depends on the network structure. Such an adjustment to the NCG was
also suggested as an open problem in~\cite{Fabri03}. Our
\emph{penalized network creation game} (PCG) is introduced in
Section~\ref{sect:model}. $\beta$ allows to level off the infinite
penalities for disconnectivity and to study the effect of the
connection requirement in the NCG on topology and social cost of Nash
equilibria. Since the cost of connected equilibria is the same as in
the NCG, we will be most interested in existence, structure, and cost
of disconnected Nash equilibria. Naturally, if $\beta$ is high, then
Nash equilibria of the PCG reveal the same properties as those of the
NCG. It is thus not surprising that a number of insights for the NCG
can be translated directly to the PCG. If $\beta$ decreases, then
properties of Nash equilibria can change. In particular, an
interesting insight gained from our structural analysis in
Section~\ref{sect:disNE} is presented in Theorem~\ref{theo:noTree}. It
shows that the prevalent tree structures of the NCG are absent in
disconnected Nash equilibria whenever $\alpha > 1$ or $\beta > 2$.

Our analysis on the existence of disconnected networks offers relevant
insight for the analysis of distributed networks with rational agents.
In many scenarios, a priori, a given set of selfish entities has no
intrinsic motivation to create a globally connected network. In
contrast, our findings indicate a peculiar absence of non-empty
disconnected stable networks, which indicate underlying incentives
that prohibit their emergence. We failed to identify any non-empty
disconnected Nash equilibrium for $\beta > 3$. In addition, structural
conditions like constant diameter in all known equilibrium topologies
for the NCG led us to conjecture that there is a constant $\beta'$
such that the empty network is the only disconnected Nash equilibrium
for any PCG with $\beta > \beta'$. This appears somewhat suprising,
because the agents in the PCG are not explicitly forced into
connection. In addition, it reveals that in terms of topology of Nash
equilibria the assumption of infinite penalties in the NCG is not a
significant drawback.

In addition, we consider the price of anarchy in
Section~\ref{sect:PoA}. There are parameter values, for which
disconnected Nash equilibria appear but the social optimum is
connected, which could lead to an unbounded price of anarchy. However,
we show that the price of anarchy in the PCG is always bounded by
$O(n)$. In addition, Theorem~\ref{theo:PoA} reveals cases with
\emph{tightness} and a matching lower bound of $\Omega(n)$. This bound
is strictly larger than any of the known bounds for the NCG. In
Section~\ref{sect:SPoA} we contrast these findings with the scenario,
in which players can play joint coordinated deviations and consider
strong equilibria. Unless $\alpha$ and $\beta$ are within a small
range, the social optimum is also a strong equilibrium (see
Theorem~\ref{theo:SE}). In Theorem~\ref{theo:SPoA} we prove that the
price of anarchy for strong equilibria is at most 4. This reveals that
in the PCG Nash equilibria can be several orders of magnitude more
costly than strong equilibria, a question which is still unsolved for
the NCG. More generally, it shows that joint and coordinated actions
of selfish agents can drastically reduce inefficiencies in selfish
network creation. Finally, Section~\ref{sect:conclude} concludes and
presents some problems for further research.
\section{The Model and Initial Results}
\label{sect:model}
The network connection game (NCG) is a tuple $(V,\alpha)$ and can be
described as follows. The set of players $V$ is the set of vertices of
a graph. Possible edges $\{i,j\} \in V \times V$ have cost $\alpha$. A
strategy $s_i$ of a player $i$ is a subset $s_i \subset
V\backslash\{v\}$ and indicates, which edges player $i$ chooses to
build. In this way a strategy vector $s$ induces a set of edges
between the players. Given a strategy vector $s$ the individual cost
for a player $i$ is
\[
c_i(s) = \alpha|s_i| + \sum_{j \neq i} dist_s(i,j),
\]
where $\alpha > 0$ and $dist_s(i,j)$ is the length of a shortest-path
in the undirected graph $G_s = (V, E_s)$ induced by the strategy
vector $s$. Note that $G_s$ is assumed to be undirected, i.e., each
edge can be traversed in any direction, independent of which player
pays for it. In the regular connection game $dist_s(i,j) = \infty$ if
players $i$ and $j$ are in different components of $G_s$. In the
\emph{penalized network creation game} (PCG) we are given a penalty
value $\beta > 1$, and $dist_s(i,j) = \beta$ for players $i$ and $j$
in different components. A pure \emph{Nash equilibrium (NE)} is a
state $s$, in which no player can unilaterally decrease his cost $c_i$
by changing his strategy $s_i$. We will restrict our attention to pure
equilibria throughout the paper. The \emph{social cost} $c(s)$ of a
state $s$ is simply $c(s) = \sum_{i \in V} c_i(s)$. A social optimum
state $s^*$ is a state with minimum social cost. Note that for the
cost of a state it does not matter, who builds an edge, and hence we
will sometimes consider the graph $G_s$ instead of $s$.
States that play an important role in the analysis of the PCG are the
empty state $s_\emptyset = (\emptyset,\ldots,\emptyset)$, $s_K$
corresponding to the complete graph, in which each edge $\{i,j\}$ with
$i\neq j$ is paid by player $\min\{i,j\}$, and $s_Z$ corresponding to
a center-sponsored star, in which one player purchases edges to all
other players.

Fabrikant et al.~\cite{Fabri03} show that there is always a pure NE in
the NCG and mention that it might be found by iterative improvement
steps. Finding a best-response for a player in a NCG, however, was
shown \classNP-hard~\cite{Fabri03}, and this translates to the PCG for
sufficiently large penalty cost. In addition, we show that
better-response dynamics may cycle, hence the game is no potential
game~\cite{Monderer96}. As the dynamics involve no disconnectivities,
the result follows directly for the PCG. Nevertheless, in the PCG
there is always a pure NE. This serves as a first insight to motivate
the further study of the properties of pure NE in the PCG.
\begin{theorem}
\label{theo:NashExist}
Every PCG has a pure Nash equilibrium, but neither NCG nor PCG are
potential games.
\end{theorem}
%
%
\begin{figure}
  \begin{center}
    \includegraphics[scale=0.4]{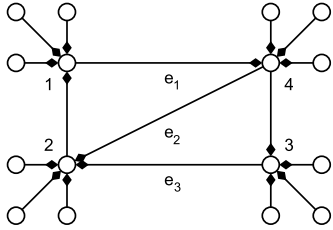}
    \caption{\label{fig:potential} A game with $k = 4$ and $4 < \alpha
      < 6$ with cycling better response iteration. Black dots indicate
      the player who pays for the edge.}
  \end{center}
\end{figure}
%
%
\begin{proof}
  We first prove the non-existence of a potential function by
  contradiction. For any $\alpha > 3$ choose an integer $k$ with $k <
  \alpha < \frac{3k}{2}$. Now construct a strategy combination for $n
  = 4k$ players as depicted in Figure~\ref{fig:potential}. The
  following steps each represent a strict improvement for the players:
  (1) player 4 removes edge $e_1$, (2) player 2 removes edge $e_2$,
  (3) player 4 builds edges $e_1$ and $e_2$. The resulting state is
  isomorphic to the initial state, in which the roles of players 2 and
  4 and edges $e_1$ and $e_3$ are switched. In particular, this allows
  us to construct an infinite improvement path, which contradicts the
  existence of a potential function.

  For the proof of existence let $\alpha \ge \beta-1$ and consider
  $s_\emptyset$. For every player a strategy change consists of
  connecting to a number $t$ of other players. As $t\alpha + t -
  t\beta = t(\alpha - (\beta-1)) \ge 0$, $s_\emptyset$ represents a
  NE. For $1 \le \alpha < \beta-1$ consider the state $s_Z$
  corresponding to a center-sponsored star, in which the center player
  pays for all edges. Using the same argument as for $s_\emptyset$ it
  is not profitable for the center player to remove any edges. For a
  leaf player connecting to additional $t$ other leaf players yields a
  difference of $t\alpha - t \ge 0$. Hence, the center-sponsored star
  represents a NE. Finally, for $\alpha < 1$ and $\alpha < \beta-1$
  consider the state $s_K$. Then, as $\alpha < 1$ every edge removal
  that leaves the graph connected cannot be profitable. The only
  possibility to disconnect the graph, however, is for player 1 to
  remove all edges. This changes his cost by $\beta(n-1) -
  (\alpha+1)(n-1) > 0$. Hence, for this case $s_K$ represents a
  NE. \qed
\end{proof}

\section{Disconnected Equilibria}
\label{sect:disNE}
In this section we consider existence and structural properties of
disconnected NE in the PCG. First, we clarify the existence of
disconnected equilibria.

\begin{theorem}
  \label{theo:DisNERange}
  For $\alpha \ge \beta-1$ the empty graph is always a disconnected
  NE.  For $0 < \alpha < \beta-1$ there is no disconnected NE.
\end{theorem}

\begin{proof}
  The first part follows from Theorem~\ref{theo:NashExist}. For the
  second part consider a player $v$ in a disconnected NE $s$. Let
  $n_v$ be the size of the component of the graph $G_s$, in which $v$
  is located. Now suppose $v$ changes his strategy by connecting to
  all $n-n_v$ players in other components. Then the change is
  $\alpha(n-n_v) + (n-n_v) - \beta(n-n_v) = (n-n_v)(\alpha - (\beta -
  1)) < 0$. Hence, under these conditions \emph{every} player in a
  disconnected state can decrease his individual cost. \qed
\end{proof}

The theorem provides a tight characterization using the empty
graph. An interesting issue, however, is to explore whether non-empty
disconnected NE are possible, because in many cases the empty graph
represents a rather unrealistic prediction for a stable network. Note
that a component of $k$ players in a disconnected NE of a PCG with
given $\alpha$ and $\beta$ must be a NE in the corresponding NCG with
$\alpha$ and $k$ players. There are a number of structures that have
been identified as components of NE in the NCG, in particular, graphs
based on affine planes (including the Petersen graph), cliques,
cliques of star graphs, and trees~\cite{Albers06}. In the following we
consider each of these classes and assume a size of at least 2
vertices to exclude the degenerate case of singleton vertices. The
treatment of affine plane graphs is cumbersome but rather
straightforward, so we omit details here. These graphs can represent
NE for $\beta < 3$ and certain restricted values of $\alpha$. More
information is available from the authors upon request.

\subsection{Cliques and Cliques of Stars} 

We refer to a pair as a component consisting of two players linked by a
single edge.

\begin{lemma}
  For $\alpha > 1$ there is no pair in a disconnected NE.
\end{lemma}

\begin{proof}
  Assume there is a pair in a disconnected NE. As one of the players
  wants to keep the edge, it must be that $\beta - \alpha - 1 \ge 0$,
  and with Theorem~\ref{theo:DisNERange} $\alpha = \beta-1$. Now
  consider a different player $v$ that constructs an edge to a player
  from the pair. The change of the individual cost is $\alpha - 2\beta
  + 3 = -\alpha + 1 < 0$. Hence, the change is profitable, which
  proves that a pair cannot appear in a disconnected NE. \qed
\end{proof}

Similarly, for larger clique components it must be $\alpha \le 1$,
because otherwise a player from the component removes one edge and
accepts the increase in distance. Theorem~\ref{theo:DisNERange} yields
the following direct corollary.

\begin{corollary}
  For $\beta > 2$ there is no clique component in a disconnected NE.
\end{corollary}

A \emph{$(k,l)$-clique of stars} is a clique with $k$ vertices, in
which each such node is the center of a star of $l$ vertices. We
consider $k \ge 3$, because otherwise the structure is a
tree. In~\cite{Albers06} it was shown that a $(k,l)$-clique of stars,
in which all edges are created by the players in the clique, is a NE
for the NCG with $\alpha = l$. Here we show that the appearance of
such a component in a disconnected NE is quite limited. 

\begin{lemma}
  \label{lem:CliqueOfStars}
  For $\alpha = l$ and $k \ge 3$, a $(k,l)$-clique of stars, in which
  all edges are created by the players in the clique, can be a
  component in a disconnected NE if and only if $\alpha = 1$ and
  $\beta = 2$.
\end{lemma}

\begin{proof}
  Consider a state $s$, which is a disconnected NE with a component
  $C$ of a $(k,l)$-clique of stars. In addition to the deviations
  considered in the NCG, for a component in a disconnected NE we must
  consider a split of the component and a connection of a vertex
  outside of the component. For $k \ge 3$ we can assign edge costs to
  the players such that no clique player is able to unilaterally
  disconnect the clique. As all edges of the stars are also built by
  players from the clique, we must have $\alpha + 1 - \beta \le 0$,
  and hence with the general bound from Theorem~\ref{theo:DisNERange}
  $\alpha = \beta - 1$.

  Consider a connection from a player outside the component to an
  arbitrary player $v$ of the clique. Hence,

  \[ \alpha - kl\beta + \sum_{w \in C} (dist_s(v,w) + 1) = \alpha
  -kl\beta + (2kl - k - l) + kl \ge 0. \]

  It is $\alpha = l = \beta - 1$, so we get as condition

  \[ 0 \le l - kl(l+1) + 3kl - k - l = -k(l^2 - 2l + 1) = -k(l-1)^2.\]

  Thus, $l = \alpha = 1$ and $\beta = \alpha + 1 = 2$. It is
  straightforward to verify that under these conditions the outlined
  state is a NE. \qed
\end{proof}

\subsection{Trees}
Tree graphs are a structure whose appearance is wide-spread in the
NCG~\cite{Fabri03,Albers06}. The following analysis shows that this
property is only due to the requirement that a NE must be
connected. In the PCG these structures can appear only in very special
cases.

\begin{lemma}
  \label{lem:oneOther}
  For $\beta > 2$ every non-singleton player $v$ in a disconnected NE
  has at least one incident edge that was created by a different
  player $w \neq v$.
\end{lemma}

\begin{proof}
  Consider a player $v$ in a component $C$ with $k$ players, who pays
  for all his $d_v$ incident edges. As we have a NE, it is not
  profitable for $v$ to disconnect from $C$, i.e., $\alpha d_v +
  \sum_{w \in C} dist(v,w) \le \beta (k-1)$. Now consider a different
  player $v'$ not in $C$ that chooses to connect to all neighbors of
  $v$. As this must not be profitable, $\alpha d_v + \sum_{w \in C}
  dist(v,w) + 2 \ge \beta k$. Adding the two inequalities yields
  $\beta \le 2$, which proves the lemma. \qed
\end{proof}

\begin{lemma}
  \label{lem:oneOther2}
  Suppose there is a disconnected NE with a component $C$ of $k > 1$
  vertices. If $\alpha > (k- 1)(\beta-2) +1$, then for every
  player $v$ there is an incident edge paid by a different player $w
  \neq v$.
\end{lemma}

\begin{proof}
  Suppose there is a player $v$ that pays for all his $d_v \geq 1$
  incident edges. As $v$ does not want to remove all edges, we have
  $\alpha d_v + \sum_{w\in C} dist(v,w) \le \beta (k-1)$, and thus
  \[
  \alpha \le \frac{1}{d_v} \left(\beta(k-1) - \sum_{w\in C}
    dist(v,w)\right).
  \]
  For every non-neighbor vertex of $C$ we have distance at least 2,
  and thus $\sum_{w \in C} dist(v,w) \ge 2(k-1) - d_v$. Substitution
  yields $\alpha \le (k-1)(\beta-2)+1$ as desired. \qed
\end{proof}

\begin{theorem}
  \label{theo:noTree}
  For $\beta > 2$ or $\alpha > 1$ no component of a disconnected NE is
  a tree.
\end{theorem}

\begin{proof}
  The first bound is a direct consequence of Lemma~\ref{lem:oneOther}
  and the fact that for a tree $|E| = |V|-1$. Thus, for disconnected
  NE with tree components $\beta \le 2$, and the second bound follows
  with Lemma~\ref{lem:oneOther2}. \qed
\end{proof}

\subsection{Non-empty Equilibria}
In the previous paragraphs we have shown that the appearance of known
NE topologies from the NCG as components in disconnected NE of the PCG
is quite limited. The existence of disconnected NE is guaranteed by
the empty network. This raises the question, under which conditions on
$\alpha$ and $\beta$ non-empty disconnected NE can evolve. We first
present a positive result.

\begin{lemma}
  \label{lem:C5}
  For $3 \le \alpha \le 4$ and $\beta \le (\alpha+11)/5$ a cycle $C_5$
  of 5 vertices can be a component of a disconnected NE.
\end{lemma}

\begin{proof}
  Consider a disconnected NE with such a component. Label the players
  in $C_5$ from 0 to 4 along a Euclidean tour. Each player $i$ pays
  for edge $(i,i+1 \mod 5)$. In this case no player can disconnect the
  component. If a player removes his edge, the increase in distance
  cost is 4, thus $\alpha \le 4$. Furthermore, every vertex in the
  cycle has a sum of distances of 6. As for each vertex $v \not\in
  C_5$ it must not be profitable to connect to a vertex $i \in C_5$,
  we get $\alpha \ge 5\beta - 11$. Every additional edge yields a cost
  improvement of at most 3, hence for $\alpha \ge 3$ it is optimal for
  $v$ to connect to at most one vertex from $C_5$. This means that for
  $\beta \le \frac{\alpha + 11}{5} \le 3$ there can be a component
  $C_5$ in a disconnected NE. \qed
\end{proof}

Note that a NE with $C_5$ is transient and not strict, i.e., there is
a sequence of strategy changes that leaves the individual cost of the
changing player identical, but leads into a non-equilibrium state. For
instance player 2 can exchange edge (2,3) by edge (2,4) without cost
change. Afterwards player 1 can strictly improve by purchasing (1,4)
instead of (1,2).

In contrast to the restricted interval, for which we can show
existence, there is an unbounded region of parameter values, for which
the empty network is the only disconnected network - in particular if
$\alpha$ or $\beta$ are large compared to $n$.

\begin{lemma}
  \label{lem:nonEmptyNE}
  In a non-empty disconnected NE let $n_l$ be the minimum size and
  $diam_l$ the minimum diameter of any non-singleton component. Then
  \begin{tabbing}
    \hspace{2.8cm} \= (1) $\alpha < 12n_l \log n_l$ \hspace{1.8cm} \= (3) $\beta < 1 + 14\sqrt{n_l \log n_l}$ \\
    \> (2) $\beta \le 1 + 2\cdot diam_l$ \> (4) if $n > 6$, then $\beta < n/2$
 \end{tabbing}
\end{lemma}
\begin{proof}
  For the first bound consider $\alpha \ge 12n_l \log n_l$ and
  a component with $n_l$ players. This component must represent a
  NE in a NCG with the same $\alpha$ and $n_l$ players, and thus
  according to~\cite{Albers06} must be a tree. This contradicts
  Theorem~\ref{theo:noTree} and the bound follows.
  
  Now consider a non-empty disconnected NE $s$ for $\beta > 2$, and
  let $C$ be a non-singleton component. As $C$ is no tree, it must
  contain at least one cycle. Let $U$ be a smallest of all cycles in
  $C$, and let $v$ be an arbitrary player that constructed some edge
  $e$ of $U$. Denote by $s'$ the state that evolves if player $v$
  removes edge $e$. Note that by this removal no additional pair of
  players gets disconnected. As $s$ is a NE, we have
  \begin{equation}
    \label{eq1}
    \alpha \le \sum_{w \in C}  (dist_{s'}(v,w) - dist_s(v,w)).
  \end{equation}
  As we have chosen $U$ to be of minimum size, all shortest distances
  between vertices of $U$ are given by the paths along the
  cycle. Thus, there is always a vertex $u$, for which the distances
  in $s$ and $s'$ are the same. This yields 
  \[ dist_{s'}(v,w) \le dist_{s'}(v,u)+dist_{s'}(u,w) =
  dist_s(v,u)+dist_s(u,w)\]
  for all $w \in C$. With $n_C = |C|$ we can conclude $\alpha \le
  2n_C\cdot diam(C) - \sum_{w \in C} dist_s(v,w)$. On the other hand,
  no vertex outside $C$ must be able to profit from a connection to
  $v$, hence $\alpha + n_C + \sum_{w \in C} dist(v,w) \ge
  n_C\beta$. The last two inequalities imply that $\beta \le 2\cdot
  diam(C)+1$ and thus deliver the second bound. As each component $C$
  must be a NE of a NCG, we know from~\cite{Fabri03} that $diam(C) \le
  \sqrt{4\alpha + 1}$. Together with the first bound on $\alpha$ shown
  above this implies the third bound $\beta < 1 + 14\sqrt{n_l \log
    n_l}$.  For the proof of the last bound the inequality (\ref{eq1})
  and the bound $\alpha + n_C + \sum_{w \in C} dist(v,w) \ge n_C\beta$
  imply
  \[ \beta \le 1 + \frac{1}{n_C} \sum_{w \in C} dist_{s'}(v,w) \]
  The sum of distances for $v$ is maximal iff $C$ in $s$ is a cycle
  and thus in $s'$ is a chain with $v$ being one of its endpoints. In
  this case
  \[ \sum_{w \in C} dist_{s'}(v,w) \le \frac{(n_C-1)n_C}{2}~. \]
  If $C$ is not a cycle in $s$, then the inequality is strict and the
  previous formulas yield $\beta < 1 + \frac{n_C-1}{2} =
  \frac{n_C+1}{2} \le \frac{n}{2}$. If $C$ is a cycle in $s$, it is
  straightforward to show that it must hold $n_C \le 5$ as $s$ is a
  NE. Here we use the assumption that $n > 6$ and get $\beta \le
  \frac{n_C+1}{2} \le 3 < \frac{n}{2}$. This proves the last
  bound. \qed
\end{proof}

In contrast to these bounds, we have not been able to derive any
non-empty disconnected NE for values of $\beta > 3$. This led us to
formulate the following conjecture.
\begin{conjecture}[Constant Penalty Conjecture]
  There is a constant $\beta'$ such that for $\beta > \beta'$ the only
  disconnected NE is $s_\emptyset$.
\end{conjecture}
Note that our bounds imply that if the conjecture is false, then there
must be non-tree NE in the NCG with a diameter in $\omega(1)$. This
seems quite unlikely, as all non-tree NE found so far have diameter at
most 3.

\section{Price of Anarchy}
\label{sect:PoA}
In this section we consider the price of anarchy in the PCG. We first
present an overview of the social optima of the game in
Figure~\ref{fig:Opt}. 

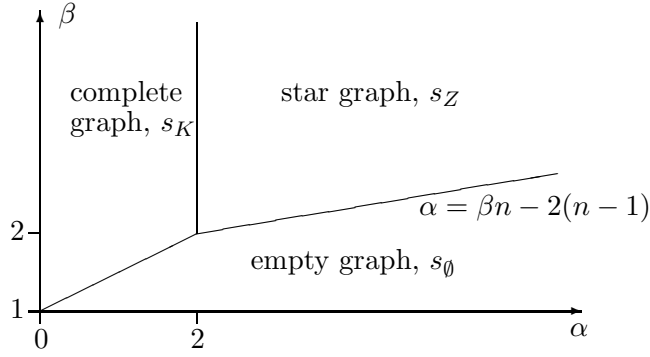
\begin{figure}
\begin{center}
\setlength{\unitlength}{.8mm}
\begin{picture}(95,60)
\put(3,5){\vector(1,0){92}}
\put(5,3){\vector(0,1){52}}
\put(31,3){\line(0,1){2}}
\put(5,5){\line(2,1){26}}
\put(3,18){\line(1,0){2}}
\put(31,18){\line(0,1){35}}
\put(31,18){\line(6,1){60}}
\put(10,40){\normalsize complete}
\put(10,35){\normalsize graph, $s_K$}
\put(40,12){\normalsize empty graph, $s_\emptyset$}
\put(45,40){\normalsize star graph, $s_Z$}
\put(68,21){$\alpha = \beta n -2(n-1)$}
\put(93,1){$\alpha$}
\put(8,53){$\beta$}
\put(4,-1){0}
\put(0,4){1}
\put(30,-1){2}
\put(0,17){2}
\end{picture}
\end{center}
\caption{\label{fig:Opt} Social optima in the PCG}
\end{figure}

\begin{theorem}
  The social optimum is $s^* = s_\emptyset$ for $2\beta - 2 < \alpha <
  2$, as well as for $\alpha \ge \max\{2, \beta n - 2(n-1)\}$. It
  holds that $s^* = s_K$ in the range $\alpha \le \min\{2, 2\beta -
  2\}$. In the remaining range the social optimum is $s^* = s_Z$.
\end{theorem}

\begin{proof}
  First, assume $\alpha \le 2$. In this case $\alpha+2 \le 4 \le 2
  dist_s(v,w)$ for every indirectly connected pair of players. Hence,
  a direct connection can only result in smaller social cost. If in
  addition $2\beta-2 \le \alpha$, then disconnectivity is the
  cheapest alternative and thus $s_\emptyset$ is optimal. Otherwise,
  for $2\beta-2 \ge \alpha$ the state $s_K$ corresponding to the
  complete graph is optimal.

  In the case $\alpha > 2$ and $\beta \le 2$ we have $2\beta-2 <
  \alpha$ and $2\beta \le 2dist_s(v,w)$ for any $dist_s(v,w) \ge
  2$. Hence, $s_\emptyset$ represents a social optimum.

  For $\alpha \ge \beta n - 2(n-1)$ and $\beta \ge 2$ consider a state
  $s$ corresponding to a non-empty graph $G_s$. Suppose $C$ is a
  non-singleton component of $G_s$ and $n_C = |C| > 1$. By $c(C)$ we
  consider only the costs introduced by pairs of players in $C$, and
  by $m_C$ the number of edges within component $C$. It is easy to
  note $c(C) \ge 2n_C(n_C-1) + (\alpha - 2) m_C$
  (see~\cite{Fabri03}). If $s$ is a candidate for a social optimum,
  component $C$ must be a star. However,
  \begin{equation}
    \label{eq:costLB}
    c(C) \ge 2n_C(n_C-1) + (\alpha - 2) m_C \ge 2n_C(n_C-1) +
    n_C(\beta - 2)(n_C-1) = n_C(n_C-1)\beta.
  \end{equation}
  Thus, by removing all edges from non-singleton components we reduce
  the social cost of the state even further. This implies that
  $s_\emptyset$ is a social optimum.

  Finally, in the case $2 < \alpha < \beta n - 2(n-1)$ the cheapest
  \emph{connected} state is $s_Z$ corresponding to the star. We show
  here that for every \emph{disconnected} state there is a strictly
  cheaper state. Hence, no disconnected state can be a social optimum,
  and the social optimum remains $s_Z$. This finishes the proof of the
  theorem.

  \begin{lemma}
    If $2 < \alpha < \beta n - 2(n-1)$, then for each disconnected
    state $s$ there is a state with strictly less social cost.
  \end{lemma}

  \begin{proof}
    Suppose for contradiction that there is a disconnected state $s$
    which is the social optimum. Similarly to the last paragraph we
    see that the components of $s$ must be stars. In addition, $s$
    must have at least one non-singleton component, because
    $s_\emptyset$ is more costly than $s_Z$:
    \[
    n(n-1)\beta > (n-1)(\alpha + 2n - 2) = \alpha(n-1) + 2(n-1)^2.
    \]
    Suppose $s$ has two (star) components $C_1$ and $C_2$ of $n_1$ and
    $n_2$ players, respectively, and assume $n_1 > 1$. We will
    consider two different states, which both must not be
    cheaper. This delivers a contradiction.

    On the one hand, suppose we remove all edges from $C_1$ and
    instead connect all players of $C_1$ to the center player of
    $C_2$. Then regarding players $v \not\in C_1 \cup C_2$ there is no
    cost change. We need one additional edge for the center player
    $v_1$ of $C_1$. Every other player from $C_1$ is now at a distance
    of 2 to $v_1$, thus there is an increase in distance of $2(n_1 -
    1)$. For the newly created distances between players of $C_1$ and
    $C_2$ we have
    \[
    \sum_{v \in C_1, w\in C_2} dist(v,w) = 2(n_2-1)n_1 + n_1,
    \]
    which is counted twice in the social cost. However, we save a
    penalty of $2n_1n_2\beta$. By assumption the total changes must
    not lead to an improved state, hence
    \[
    \alpha + 2(n_1-1) + 4(n_2-1)n_1 + 2n_1 - 2n_1n_2\beta \ge 0.
    \]
    Thus, $\alpha \ge 2-4n_1n_2 + 2n_1n_2\beta$. On the other hand,
    consider the change in cost when we remove all edges from
    $C_1$. Again, by assumption this must not lead to a cheaper state,
    so
    \[
    \alpha(n_1 - 1) + 2(n_1 - 1)^2 \le n_1(n_1 - 1)\beta.
    \]
    This implies $\alpha \le n_1\beta - 2(n_1 - 1)$. Combining the
    bounds leads to $n_1(\beta - 2) + 2 \ge 2 - 4n_1n_2 +
    2n_1n_2\beta$, which implies $4n_2 - 2 \ge \beta(2n_2 - 1)$ and
    thus $\beta \le 2$, which is a contradiction to $2 < \beta n -
    2(n-1)$. \qed
  \end{proof}
\end{proof}

For $\alpha < \beta-1$ we have seen in Theorem~\ref{theo:DisNERange}
that no disconnected NE exists. In addition, the next theorem shows
that in this case a finite penalty for disconnectivity cannot disrupt
any NE of the NCG. Hence, for this parameter range the price of
anarchy is identical to the NCG.
\begin{theorem}
  For $\alpha < \beta-1$ the NE of the PCG are exactly the NE of the
  NCG, so in this parameter range the prices of anarchy and stability
  remain the same in the PCG as in the NCG, respectively.
\end{theorem}
\begin{proof}
  As in this parameter range all NE of the PCG must be connected,
  every such state must also be a NE in the NCG, because in the PCG
  the players only consider \emph{more} potentially profitable
  deviations. For the converse, suppose for contradiction that there
  is a NE $s$ of the NCG, which is not a NE in the PCG. As $s$ must be
  connected, there is a player $v$, who profits from changing her
  strategy $s_v$ to a strategy $s'_v$ that disconnects the resulting
  graph. Let $W$ be the set of other players, to which $v$ is
  disconnected under $s' = (s'_v,s_{-v})$. Now suppose $v$ changes his
  strategy again to $s''_v$ by building direct connections to all
  players from $W$. As $\alpha + 1 < \beta$, this is again a
  profitable deviation. Thus, player $v$ strictly profits from
  switching from $s_v$ to $s''_v$. Note, however, that in $s'' =
  (s''_v, s_{-v})$ the resulting graph is connected, thus $s''_v$ must
  be a feasible deviation yielding profit for $v$ in the NCG. This
  contradicts the assumption that $s$ is a NE in the NCG. \qed
\end{proof}
In general, however, the price of anarchy for the PCG can be strictly
larger than for the NCG. Figure~\ref{fig:PoA} provides an overview of
the bounds we obtained. Note that all these bounds are in $O(n)$ for
the respective parameter values.
\begin{figure}[ht]
\begin{center}
\setlength{\unitlength}{1mm}
\begin{picture}(105,73)
\put(5,7){\vector(1,0){92}}
\put(7,5){\vector(0,1){62}}
\put(17,5){\line(0,1){12}}
\put(5,12){\line(1,0){2}}
\put(5,17){\line(1,0){2}}
\put(7,7){\line(1,1){60}}
\put(17,12){\line(6,1){80}}
\put(17,12){\line(2,1){80}}
\put(49,28){\line(0,1){21}}
\put(49,5){\line(0,1){2}}
\put(77,5){\line(0,1){37}}
\put(2,6){1}
\put(2,11){2}
\put(2,16){3}
\put(6,1){0}
\put(16,1){2}
\put(46,1){$\Theta(n)$}
\put(70,1){$12n \log n$}
\put(95,3){$\alpha$}
\put(9,64){$\beta$}
\put(11,9){\normalsize $< 2$}
\put(42,10){\normalsize $o(n^\epsilon)$}
\put(48,24){\normalsize $o(n^\epsilon) + O\left(\frac{n\beta}{n + \alpha}\right)$}
\put(87,12){\normalsize 1}
\put(83,32){\normalsize $\Theta\left(\frac{n\beta}{\alpha}\right)$}
\put(63,45){\normalsize $\Theta(n)$}
\put(30,25){\normalsize $\Theta(\beta)$}
\put(12,45){\normalsize same as in NCG}
\put(12,40){\normalsize see Fig.~\ref{fig:NCG}}
\put(64,60){\tiny $\beta -1$}
\put(81,48){\tiny $2\beta -2$}
\put(81,20){\tiny $n\beta - 2(n-1)$}
\end{picture}
\end{center}
\caption{\label{fig:PoA} Price of anarchy in the PCG}
\end{figure}
Our proof is divided into ranges, in which different network
structures are social optima.

\subsection{Star Graph}
At first, we concentrate on the case $\max\{2, \beta-1\} < \alpha <
\beta n - 2(n-1)$, in which disconnected NE can appear and the star is
the social optimum.

We start by observing a helpful reduction to the price of anarchy in
the NCG. Consider a disconnected NE $s$ with non-singleton components
$C_1,\ldots,C_r$ and singleton components
$C_{r+1},\ldots,C_{r+l}$. Let $n_i = |C_i|$ and $c(C_i)$ be the cost
of $C_i$ as a NE in a NCG with $n_i$ players. In particular, $c(C_i)$
counts only edge and shortest path costs within $C_i$ but no
penalties. In addition, let $s_{Z_i}$ be a social optimum state for a
NCG with $n_i$ players. For the given parameter range of $\alpha$ all
the $s_{Z_i}$ are stars.

\begin{lemma}
  \label{lem:compoPoA}
  For the cost of $s$ it holds that
  \[
  \frac{c(s)}{c(s_Z)} \le \frac{n \beta}{\alpha + 2(n-1)} + \max_{1 \le i \le
    r} \left\{\frac{c(C_i)}{c(s_{Z_i})}\right\}~.
  \]
\end{lemma}

\begin{proof}
Obviously it holds that
\[
c(s) = 2\beta \left(\sum_{1\le i < j \le r+l} n_i n_j\right) +
\sum_{i=1}^r c(C_i) \le \beta n(n-1) + \sum_{i=1}^r c(C_i).
\]
With $n = l + \sum_{i=1}^r n_i$ we can bound as follows

\begin{eqnarray*}
\frac{\sum_{i=1}^r c(C_i)}{c(s_Z)} & = & \sum_{i=1}^r \left(\frac{c(C_i)c(s_{Z_i})}{c(s_{Z_i})c(s_Z)}\right) 
\le \max_{1 \le i \le r} \left\{\frac{c(C_i)}{c(s_{Z_i})}\right\} \cdot \sum_{i=1}^r\frac{c(s_{Z_i})}{c(s_Z)}\\
& = & \max_{1 \le i \le r} \left\{\frac{c(C_i)}{c(s_{Z_i})}\right\} \cdot \sum_{i=1}^r\frac{(n_i-1)\alpha + 2(n_i - 1)^2}{(n-1)\alpha + 2(n-1)^2}\\
& = & \max_{1 \le i \le r} \left\{\frac{c(C_i)}{c(s_{Z_i})}\right\} \cdot \frac{\alpha \sum_{i=1}^r (n_i-1) + 2\sum_{i=1}^r(n_i - 1)^2}{(n-1)\alpha + 2(n-1)^2}\\
& < & \max_{1 \le i \le r} \left\{\frac{c(C_i)}{c(s_{Z_i})}\right\}~,
\end{eqnarray*}
which proves the lemma. \qed
\end{proof}

\begin{theorem}
  \label{theo:PoA}
  For $2\beta - 2 \le \alpha \le n\beta - 2(n-1)$ the price of anarchy
  is bounded by
  \[ \Theta\left(\frac{n\beta}{\alpha}\right) \mbox{ for } \alpha \ge
  12n \log n, \hspace{1.5cm} O\left(5^{\sqrt{\log n}}\log n + \frac{n
      \beta}{\alpha + n}\right) \mbox{ for } \alpha < 12n \log n~.\]
  For $\beta - 1 \le \alpha \le 2\beta - 2$ the price of anarchy is
  $\Theta(\min\{\beta,n\})$.
\end{theorem}

\begin{proof}
  For the proof of the first bound consider $\alpha \ge 12n \log
  n$. According to Lemma~\ref{lem:nonEmptyNE} in this case every NE is
  either connected or $s_\emptyset$. For $\alpha \ge 12n \log n$ all
  connected NE have a constant price of anarchy~\cite{Albers06}, while
  $s_\emptyset$ leads to an increase and proves our first bound:
  \begin{equation}
    \label{eq:emptyCost}
    \frac{c(s_\emptyset)}{c(s_Z)} = \frac{\beta n}{\alpha +
      2(n-1)} \in \Theta\left(\frac{n\beta}{\alpha}\right)~.
  \end{equation}
  This bound increases from $\Theta(1)$ to $\Theta(n)$ if $\alpha$
  drops from $n\beta - 2(n-1)$ to $2\beta - 2$. It also shows that
  the price of anarchy induced by $s_\emptyset$ is never more than
  $O(n)$ for $s^* = s_Z$ and $\alpha \ge \beta - 1$. Another range,
  for which $s_\emptyset$ is the most expensive NE, is $\beta-1 \le
  \alpha \le 2\beta - 2$ with $\beta \ge 7$. Then any directly
  connected pair induces a cost of $\alpha + 2 \le 2\beta$. Any
  indirectly connected pair in a NE induces a cost $2dist_s(v,w) \le
  2\sqrt{4\alpha + 1} \le 2\sqrt{8\beta-7} \le 2\beta$. Thus, the cost
  of $2\beta$ induced by $s_\emptyset$ is maximal for every pair of
  players. Therefore the fraction in Equation~(\ref{eq:emptyCost})
  characterizes the price of anarchy and results in
  $\Theta(\min\{\beta,n\})$, which proves the third bound.

  For the remaining range with $\alpha < 12n \log n$ we cannot exclude
  the possibility that there are worse disconnected NE than
  $s_\emptyset$. However, components of these NE must be connected NE
  of smaller NCGs. Using Lemma~\ref{lem:compoPoA} we can bound the
  price of anarchy for these NE by the sum of the fraction for
  $s_\emptyset$ in Equation~(\ref{eq:emptyCost}) plus the maximum
  factor of any component NE in the corresponding NCG. With the bound
  of $5^{\sqrt{\log n}}\log n \in o(n^\epsilon)$ on the price of
  anarchy for the NCG~\cite{Demaine07} this proves our second bound
  $O\left(5^{\sqrt{\log n}}\log n + \frac{n\beta}{\alpha + n}\right) =
  O(\max\{5^{\sqrt{\log n}}\log n, \min\{n, \beta\}\})$. In
  particular, this represents a bound of $O(n)$ for the price of
  anarchy. \qed
\end{proof}

\subsection{Complete Graph}

In this case we have $s^* = s_K$, and thus it must hold $\beta-1 \le
\alpha \le \min\{2,2\beta-2\}$. The following theorem summarizes the
bounds.

\begin{theorem}
\label{theo:KnPoA}
The price of anarchy is bounded by $4/3$ for $\alpha < 1$, $4/3$ for
$1 \le \alpha \le 2$ and $\beta < 2$, and $3/2$ for $\alpha <
\min\{2,2\beta-2)\}$ and $\beta \ge 2$.
\end{theorem}

\begin{proof}
  Suppose $s$ is a NE and let $C$ be a component. For any two players
  $v,w \in C$ we have $dist_s(v,w) \le \alpha + 1$, because otherwise
  building a direct connection is profitable. Thus, for $\alpha < 1$
  there are no indirectly connected players. A directly connected pair
  of players yields a social cost of $\alpha + 2 \le 2\beta$. Hence,
  the worst NE is $s_\emptyset$ and the price of anarchy is bounded by
  \[
  \frac{c(s_\emptyset)}{c(s_K)} = \frac{2\beta}{\alpha + 2} \le
  \frac{2\alpha + 2}{\alpha + 2} < \frac{4}{3}~.
  \]
  This proves the first bound. For $1 \le \alpha < 2$, the diameter is
  $diam(C) \le 2$ for every non-singleton component of a NE
  $s$. Indirectly connected players yield a social cost of $2 diam(C)
  \le 4$ and directly connected players $\alpha + 2 < 4$. If $\beta <
  2$, then $2\beta < 4$ for any disconnected pair of players. Thus,
  the price of anarchy is bounded by
  \[
  \frac{4(n(n-1)/2)}{c(s_K)} = \frac{4}{\alpha + 2} \le \frac{4}{3}~,
  \]
  which proves the second bound. In case $\beta \ge 2$, we get $2\beta
  \ge 4 \ge \alpha + 2$, so $s_\emptyset$ is the worst NE. For the
  price of anarchy we get the third bound by
  \[
  \frac{c(s_\emptyset)}{c(s_K)} \le \frac{2\alpha + 2}{\alpha + 2} <
  \frac{3}{2}~.
  \]
\qed
\end{proof}
\subsection{Empty Graph}

In this case we have $2\beta - 2 < \alpha < 2$ or $\alpha \ge
\max\{2, \beta n - 2(n-1)\}$. The following theorem summarizes the
bounds.

\begin{theorem}
  The price of anarchy is bounded by $3/2$ for $2\beta - 2 < \alpha <
  1$, $2$ for $1 \le \alpha < 2$ and $\alpha > 2\beta - 2$, and $1$
  for $\alpha \ge 12n \log n$ and $\alpha > \beta n - 2(n-1)$. It is
  $O(5^{\sqrt{\log n}}\log n\cdot\frac{\alpha + n}{n\beta})$ for the
  remaining range.
\end{theorem}

\begin{proof}
  In the range $2\beta-2 < \alpha < 1$ every component of a NE
  is a clique. Every directly connected pair of players yields a
  contribution to the social cost of $\alpha + 2 > 2\beta$. For $n >
  2$ we can assign the edges of a complete graph to be purchased by
  the players such that no player can disconnect the graph by removing
  his edges. Hence, the complete graph represents the worst NE and with
  \[
  \frac{c(s_K)}{c(s_\emptyset)} = \frac{\alpha + 2}{2\beta} <
  \frac{3}{2}
  \]
  the first bound follows. 

  In the range $1 \le \alpha < 2$ and $\alpha > 2\beta-2$ every
  component has a diameter of at most 2. A connected pair of players
  yields a social cost of at most $4 > 2\beta$ or $2 + \alpha >
  2\beta$. Therefore the worst NE is connected, has diameter at most
  2, and as $4 > 2 + \alpha$ as few edges as possible. This means no
  NE can be more costly than the star graph (which does not represent
  a NE here). The price of anarchy is bounded by
  \[
  \frac{c(s_Z)}{c(s_\emptyset)} = \frac{\alpha(n-1)+2(n-1)^2}{\beta n(n-1)} < 2~,
  \]
  which proves the second bound.

  Note that for $\alpha \ge 12n \log n$ Lemma~\ref{lem:nonEmptyNE}
  shows that NE can only be connected or
  empty. Lemma~\ref{lem:oneOther2} then shows for $\alpha >
  n\beta - 2(n-1)$ that $s_\emptyset$ is the only NE in this
  range. Hence, we get a price of anarchy of 1.

  For the remaining range of $2 \le \alpha < 12n \log n$ and $\alpha >
  \beta n - 2(n-1)$ we use a bounding argument over the components of
  a disconnected NE. Similar as for the price of anarchy consider a
  disconnected NE $s$ with non-singleton components $C_1,\ldots,C_r$
  and singleton components $C_{r+1},\ldots,C_{r+l}$. Let $n_i = |C_i|$
  and $c(C_i)$ be the cost of $C_i$ as a NE in a NCG with $n_i$
  players. In particular, $c(C_i)$ counts only edge and shortest path
  costs within $C_i$ but no penalties. In addition, let $s_{Z_i}$ be a
  state for a NCG with $n_i$ players representing a star graph. Then
  Lemma~\ref{lem:compoPoA} tells us that
  \[
  \sum_{i=1}^r \frac{c(C_i)}{c(s_Z)} < \max_{1\le i \le r} \left\{ \frac{c(C_i)}{c(s_{Z_i})} \right\}.
  \]
  Similarly to Lemma~\ref{lem:compoPoA} we can bound $c(s) \le \beta
  n(n-1) + \sum_{i=1}^r c(C_i)$, and thus get the third bound
  \begin{eqnarray*}
    \frac{c(s)}{c(s_\emptyset)} & \le & 1 + \frac{c(s_Z)}{c(s_\emptyset)} \sum_{i=1}^r \frac{c(C_i)}{c(s_Z)} < 1 + \frac{c(s_Z)}{c(s_\emptyset)} \max_{1\le i \le r} \left\{ \frac{c(C_i)}{c(s_{Z_i})} \right\} \\
    & \in & O\left(5^{\sqrt{\log n}}\log n \cdot \frac{\alpha + n}{n\beta}\right)~.
  \end{eqnarray*}
  Note that by restriction to $\alpha < 12n \log n$ this bound is still in
  $o(n^\epsilon)$. \qed
\end{proof}

\section{Strong Equilibria}
\label{sect:SPoA}
In this section we assume agents are able to jointly deviate to
different strategies. As stability concept we consider the strong
equilibrium~\cite{Aumann59}, in which no coalition $C$ of players can
decrease the cost for each of its members by taking a joint
deviation. More formally, if a state $s$ is a strong equilibrium (SE),
then for each coalition of players $C$ and each possible strategy
profile $s'_C$ for the players in $C$ it holds that if there is a
player $i \in C$ with $c_i(s'_C,s_{-C}) < c_i(s)$, then there is
another player $j \in C$ with $c_j(s'_C,s_{-C}) \ge c_j(s)$. The price
of anarchy for SE is a straightforward adaption of the price for
NE. It was studied before in~\cite{Andelman07} for the NCG. The
following theorem shows that
with the exception of a small range of parameter values strong
equilibria always exist in the PCG.
%
\begin{theorem}
  \label{theo:SE}
  For $\alpha < \beta - 1$ the SE of the PCG are exactly the SE of the
  NCG. For $\alpha \ge \beta-1$ the social optimum is a SE for all
  parameter values except $\beta < 3$, and $\beta n - 2n + 2 -
  (\beta-1) < \alpha < \beta n - 2n + 2$.
%
\end{theorem}

\begin{proof}
  The first part of the theorem can be proven directly along the lines
  of Theorem~\ref{theo:DisNERange}. Consider a SE $s$ of the NCG and
  suppose there is a profitable deviation of a coalition $C$ that
  creates a disconnected graph. Then reconnecting all players across
  components creates a connected deviation that is cheaper for every
  player. This is a contradiction to $s$ being a SE.

  For the second part, we first note that for the case $\alpha = \beta
  - 1$ we can use the arguments of Theorem~\ref{theo:DisNERange} to
  show that every SE of the NCG is also a SE of the PCG. Hence, using
  results from~\cite{Andelman07} it holds that with the exception of
  $\alpha \in (1,2)$ (respectively $\beta \in (2,3)$) the social
  optimum is a SE. For the remainder we thus focus on the range
  $\alpha > \beta-1$.

  We will at first concentrate on the case, in which $s_\emptyset$ is
  a social optimum. In addition, we assume $\beta \le \min\{2,
  \frac{\alpha}{2}+1\}$ holds. Let us consider a deviating coalition
  $C$ of $n_C$ players that builds a connected component, which is
  disconnected from the remaining $n-n_C$ players. In $C$ there must
  be at least one player $v$ that pays for at least $d_v/2$ (i.e.,
  half of his incident) edges. It requires an easy inductive argument
  to show that if such a player $v$ does not exist, not all edges of
  $C$ are being paid for. For such a player we get a cost of
  \begin{eqnarray*}
    &     & d_v\alpha/2 + d_v + 2(n_C-1-d_v) + \beta(n-n_C) \\
    & \ge & \beta d_v + 2(n_C-1-d_v) + \beta(n-n_C)\\
    & \ge & \beta(n-n_C+d_v) + \beta(n_C-1-d_v) = \beta(n-1).
  \end{eqnarray*}
  Hence, player $v \in C$ is not able to strictly decrease his cost.

  For the remaining range of $\beta > 2$ and $\alpha \ge \beta n -
  2(n-1)$ in which $s_\emptyset$ is optimal, we consider a similar
  argument. Suppose a coalition $C$ of $n_C$ players builds a connected
  network. Then the cost of this network can be lower bounded as in
  Equation~(\ref{eq:costLB}). Hence, the new average player cost in
  $C$ with respect to the coalition is at least $\beta(n_C - 1)$,
  which is exactly the cost of each player in $s_\emptyset$ with
  respect to players in $C$.
  This proves that whenever $s_\emptyset$ is optimal, it is a SE.

  For the case, in which $s_K$ is optimal, it is
  known~\cite{Andelman07} that there is no connected SE for $\alpha
  \in (1,2)$ and $n \ge 7$. If $\alpha \in (1,2)$ and $\alpha/2 + 1 <
  \beta < \alpha + 1$, then $s_K$ is the unique social optimum. This
  means that for $\beta \in (1.5, 3)$ a social optimum might not be a
  SE.

  For the remainder let us consider the range of $\beta \ge 3$. For
  $\alpha \le 2$ the game is equivalent to a NCG, so $s_K$ is a SE for
  $\alpha \in [0,1]$ and $\alpha = 2$. Thus, we concentrate on the
  case $\alpha > 2$ and $\beta < \alpha+1$, in which the star $s_Z$ is
  the only social optimum and is not guaranteed to be a SE by previous
  arguments. Suppose the star is periphery-sponsored, i.e., each leaf
  vertex pays for the incident edge. Then the star center will never
  participate in a deviation: As $\beta > 2$, disconnecting a player
  can only increase the cost for him, and w.r.t. any connected
  component he can never achieve a better cost. Hence, we focus on the
  leaf players. If a coalition $C$ of $n_C$ leaf players chooses to
  deviate, it cannot find a profitable deviation that leaves the
  network connected. This is a result from the fact that the
  periphery-sponsored star is a SE in the NCG~\cite[Theorem
  4.1]{Andelman07}. Hence, let us consider a deviating coalition that
  builds a connected component $C$, which is disconnected from the
  remaining $n-n_C$ players.

  \begin{description}
  \item[Case 1:] First, suppose in $C$ there is a player that pays for
    an edge and has degree 1. For this player a lower bound on his
    cost is given by $\alpha + 1 + 2(n_C-2) + \beta(n-n_C)$. If the
    deviation is profitable for the coalition, we must have
    \[ \alpha + 2n_C - 2 + \beta(n-n_C) < \alpha + 2n - 3, \]
    because otherwise the player would refuse to join. This gives
    $\beta < 2$, which contradicts $\beta \ge 3$. 
    
  \item[Case 2:] Otherwise, suppose each player that pays for an edge
    in $C$ has degree at least 2. We again consider a player $v$, who
    pays for at least $d_v/2$ edges. This player pays for at least
    one edge and has cost at least $d_v\alpha/2 + d_v + 2(n_C -
    1 - d_v) + \beta(n-n_C)$. The player must be motivated to join
    the coalition, and hence his cost must decrease:
    \begin{eqnarray*} 
      d_v\alpha/2 + d_v + 2(n_C - 1 - d_v) + \beta(n-n_C) & < & \alpha + 2n - 3 \\
      d_v\alpha - 2d_v + 4n_C + 2\beta(n-n_C) & < & 2\alpha + 4n - 2 \\
      d_v(\alpha - 2) + 2\beta(n-n_C) & < & 2\alpha + 4(n-n_C)- 2
    \end{eqnarray*} 
    With $d_v \ge 2$ and $n_C \le n-1$ by assumption we have $\beta <
    2 + \frac{1}{n-n_C} \le 3$. 
    This upper bound is again tight. Consider a game with $n=5$
    players, in which the four leaf players deviate to a cycle. This
    deviation can be profitable for any $\beta < 3$ and appropriate
    values of $\alpha$.
  \end{description}

  Now consider deviations of a coalition of leaf players to the empty
  network. This can be profitable if $\beta(n-1) <
  \alpha+2n-3$. Together with the optimality bound of $s_Z$ this
  yields $\beta(n-1) - 2n + 3 < \alpha < \beta n - 2n + 2$, the
  second bound.

  At last, consider deviations in which a part of the coalition $C$
  builds disconnected components, while another part of the coalition
  possibly remains connected to the star of players outside $C$. For
  the remaining range, in which the star $s_Z$ can be a SE, we can use
  the above arguments to show that there must be a player in a newly
  created component that is not able to strictly decrease his
  cost. This shows that $s_Z$ is indeed a SE in the remaining range.
  \qed
\end{proof}

In combination with results from~\cite{Andelman07} the theorem shows
for the case $\alpha < \beta - 1$ that the price of anarchy for SE is
strictly larger than 1, but at most 2. The main result in this section
is a general constant upper bound on the price of anarchy for SE in
the PCG. 

\begin{theorem}
  \label{theo:SPoA}
  The price of anarchy for SE in the PCG is at most 4.
\end{theorem}

\begin{proof}
  In case the complete graph is the social optimum,
  Theorem~\ref{theo:KnPoA} shows that the price of anarchy is at most
  1.5. As the SE are a subset of the NE of a game, the theorem follows
  for this case.

  We next show the bound for the empty network as optimum. Suppose $s$
  is a non-empty SE, and consider any component $C_i$ of $s$ with $n_i
  = |C_i| > 1$. Each player that pays for at least one edge in $C_i$
  has cost at least $\alpha + n_i - 1$. A joint deviation of this set
  of players would be to delete all edges, which would result in a
  cost of $\beta (n_i-1)$ for each of them w.r.t. the players in
  $C_i$. Hence, it must be that $\alpha + n_i - 1 \le \beta (n_i -
  1)$.

  Note that connection and distance costs within components can be
  bounded as follows. Each non-empty component $C_i$ has at least one
  vertex $v_i$ with distance cost at most $\alpha + 2n_i - 3$, because
  otherwise the whole component could deviate jointly to the star and
  all decrease their cost w.r.t. vertices from $C_i$. Thus, similarly
  to~\cite{Albers06} and~\cite[Lemma 4.1]{Andelman07} we can bound the
  distance and edge costs within $C_i$ by
  \begin{eqnarray*}
  & & (n_i - 1)\left(2\alpha + n_i - 1 + \sum_{v_j \in C_i, v_j \neq v_i} dist(v_i,v_j)\right) \\
  & \le & 2\alpha(n_i - 1) + (n_i - 1)^2 + n_i(\alpha + 2n_i - 3) \\
  & \le & 3n_i \beta(n_i - 1) - 2n_i - 2\alpha + 1 < 3n_i \beta(n_i - 1) 
  \end{eqnarray*}

  In addition, for each vertex in $C_i$ there are penalties of $\beta
  \sum_{j \neq i} n_j$. This yields 
  \[
  \frac{c(s)}{c(s_\emptyset)} \le \frac{\sum_i 3n_i \beta (n_i - 1) +
    \beta n_i \sum_{j \neq i} n_j }{\beta n(n-1)} \le 1 + 3n \max_i
  \frac{ n_i - 1}{n (n-1)} \le 4,
  \]
  and proves the bound for the empty network as social optimum. 

  If the star is the social optimum, then for each connected SE the
  price of anarchy for SE is at most 2. Consider a disconnected SE
  with $k$ components and number the components such that $n_1 \ge n_2
  \ge ... \ge n_k$. We can bound the price of anarchy for SE by the
  maximum factor achieved by any component (see
  Lemma~\ref{lem:compoPoA}) in addition to the costs incurred by the
  penalties. As each component must represent a SE, we have
  \[ \frac{c(s)}{c(s_\emptyset)} \le 2 + \frac{\sum_i n_i \beta \sum_{j
      \neq i} n_j}{\alpha(n-1) + 2(n-1)^2}. \]
  For the remaining part we focus on the penalties. Each player in
  component $C_i$ has penalty exactly $\beta (n-n_i)$. On the other
  hand, if all players join and create an additional star, then his
  new cost for players outside $C_i$ is at most $\alpha + 2(n-n_i)$
  when being a leaf. This yields
  \[ \beta (n-n_1) \le \beta (n-n_i) \le \alpha + 2(n-n_i) \le \alpha
  + 2(n-1).\]
  Therefore, $\alpha \ge \beta(n-n_1) - 2(n-1)$, which allows us
  to bound
  \begin{eqnarray*}
    \frac{\sum_i n_i \beta \sum_{j\neq i} n_j}{\alpha(n-1) + 2(n-1)^2} 
    & \le & \frac{\sum_i\beta n_i (n - n_i)}{\beta(n-1)(n-n_1)} 
  \end{eqnarray*}
  Suppose that $r := n-n_1$, then $n_1 = n-r$ and the number of
  non-connected vertex pairs is $(n-r)r$ for the pairs involving $C_1$
  and at most $r^2-r$ for the remaining components. Hence,
  \begin{eqnarray*}
    \frac{\sum_i n_i \beta \sum_{j\neq i} n_j}{\alpha(n-1) + 2(n-1)^2} 
    & \le & \frac{2\beta ((n-r)r + r^2 - r)}{\beta(n-1)r} = \frac{2\beta (n-1)}{\beta(n-1)} = 2
  \end{eqnarray*}
  which proves the theorem for the star network as social
  optimum. \qed
\end{proof}

\section{Conclusions}
\label{sect:conclude}
In this paper we have extended a model for selfish network creation to
allow for finite penalty values for disconnectivity. Our analysis of
the resulting game and disconnected Nash equilibria brings up a number
of interesting insights. All Nash (strong) equilibria of the NCG can
be Nash (strong) equilibria for the penalized game under sufficiently
large penalty values. Tree structures do almost never appear in
disconnected Nash equilibria. There are cases in which the price of
anarchy is $\Theta(n)$ and thus strictly higher than in the NCG. In
contrast, the strong price of anarchy remains a constant and is at
most 4. However, the increase for the price of anarchy is due to the
existence of the empty network as a Nash equilibrium. Once we can
exclude emptiness of the network, we conjecture that above a constant
threshold for the penalty no disconnected non-empty Nash or strong
equilibrium exists. This would mean that for all non-empty
disconnected Nash equilibria the price of anarchy is similar to the
NCG also bounded by $o(n^\epsilon)$. Proving or disproving this
conjecture remains as an interesting open problem. In addition, it
would be interesting to observe similar phenomena in models with
different edge costs, e.g. given by hierarchical metrics as
in~\cite{Kleinberg08}. In general, deriving a deeper understanding of
the properties and structural characterizations of Nash and strong
equilibria in network creation games like the PCG is an interesting
research direction.

\bibliographystyle{plain}
\bibliography{../../../bibfiles/price,../../../bibfiles/mechanism,../../../bibfiles/steiner,../../../bibfiles/socecon,../../../bibfiles/game,../ref}

\end{document}